\documentclass[11pt]{article}

\usepackage[a4paper,margin=1.85cm]{geometry}
\usepackage{amsmath,amssymb,amsthm,mathtools,bm}
\usepackage{braket}
\usepackage{booktabs,tabularx,array,multirow}
\usepackage{graphicx,float}
\usepackage{xcolor,microtype}
\usepackage{enumitem}
\usepackage[authoryear,round]{natbib}
\usepackage{hyperref}

\hypersetup{colorlinks=true,linkcolor=blue!55!black,citecolor=blue!55!black,urlcolor=blue!55!black}
\setlist{leftmargin=2em,itemsep=2pt,topsep=3pt}
\newtheorem{theorem}{Theorem}[section]
\newtheorem{proposition}[theorem]{Proposition}
\newtheorem{lemma}[theorem]{Lemma}

\theoremstyle{remark}
\newtheorem{remark}[theorem]{Remark}
\newcommand{\C}{\mathbb C}

\newcommand{\tr}{\operatorname{tr}}
\newcommand{\rank}{\operatorname{rank}}

\DeclarePairedDelimiter{\norm}{\lVert}{\rVert}
\DeclarePairedDelimiter{\abs}{\lvert}{\rvert}

\title{Gram-Certified Resource Continuation for Structured Quantum Representation Audits}
\author{%
Azadeh Alavi$^{1,2,*\dagger}$, Fatemeh Kouchmeshki$^{2,\dagger}$,\\
Hossein Akhoundi$^{2,\dagger}$, and Abdolrahman Alavi$^{2}$\\[2mm]
\small $^1$School of Computing Technologies, RMIT University, Melbourne, VIC, Australia\\
\small $^2$Pattern Recognition Pty Ltd, Melbourne, VIC, Australia\\
\small $^*$Correspondence: \href{mailto:azadeh.alavi@rmit.edu.au}{azadeh.alavi@rmit.edu.au}\\
\small ORCID (A.Az.): \href{https://orcid.org/0000-0002-9565-217X}{0000-0002-9565-217X}\\
\small $^\dagger$These authors contributed equally.}
\date{}

\begin{document}
\maketitle

\begin{abstract}
Dense representation of an $n$-qubit pure state requires $2^n$ complex
amplitudes, precluding dense classical materialization at large $n$.  We
develop Gram-certified resource continuation for structured
quantum-representation workloads and ask when a solution obtained under a
lower-cost resource model remains a justified initialization for a richer one.
For coarse and fine state ensembles connected by a declared isometry, fine,
coarse, and cross complex amplitude overlaps form a positive-semidefinite
block Gram matrix.  A signed operator of dimension at most twice the sample
count has the nonzero signed spectrum of the fine density minus the lifted
coarse density, yielding trace- and operator-norm diagnostics without
constructing either density operator.  We prove that a coarse weighted
spectral flag with objective suboptimality $\delta_c$ has fine-level
suboptimality at most $\delta_c+2\varepsilon$, where $\varepsilon$ is the
empirical trace distance; the factor two is attainable.  We distinguish
encoder change from exact feasible-family prolongation, give a gap-dependent
subspace-stability test, and show that continuation cannot overcome a final
Schmidt-rank ceiling.  In deterministic synthetic controls over an
8-to-40-qubit ladder, exact ancilla lifts agree to numerical precision.
Transferred initialization reduces final-rung block updates from 30 to 20,
but the complete cascade costs $4.80$--$5.43$ times a direct final-rung solve,
without material objective improvement.  Reordering eight Bell pairs reduces
the maximum matrix-product-state bond from 256 to 2.  Thus continuation is
justified only when cross-rung mismatch, feasible-family inclusion, topology,
and total work jointly satisfy prespecified audits.  These noise-free
classical results neither establish generic 40-qubit simulability nor claim
hardware performance or quantum advantage.
\end{abstract}


\section{Introduction}

\citet{alavi_invariance_2026} identified noiseless fidelity kernels and
variational-return scores with Hermitian projector, density, Grassmann, and
flag geometry.  That invariance audit answers
\emph{what object is being compared and what information it removes}.
Building on that foundation, amplitude-Gram and bounded tensor contractions
permit empirical spectral flags to be evaluated without a dense
Hilbert-space operator.  All reported quantities are obtained from
deterministic, noise-free structured-state calculations on classical
hardware, without a physical quantum device.  The certificates identify when
a declared cross-rung transfer preserves the empirical objective, conditional
on access to the required structured contractions and complex overlaps.  The present paper
addresses the next computational question: \emph{when is a solution obtained
under a lower-cost resource model a scientifically justified initialization for a
richer one?}

The answer cannot be ``whenever parameters have the same names.''  Increasing
qubit width can change the feature map and empirical density.  Increasing a
tensor rank crosses strata of a bounded-rank algebraic variety.  Changing site
order or tree topology can alter the relevant cut ranks exponentially.  Even
when final-rung iterations decrease, the earlier rungs may make a cascade more
expensive than a direct solve.  We therefore replace informal warm starting by
an audited directed system of problems, isometries, feasible-family
inclusions, and rejection certificates.

Our contributions are:
\begin{enumerate}
  \item a trace-distance transfer theorem for multirank weighted flags and an
  exact $2m$-dimensional amplitude-Gram realization of its certificate;
  \item objective-preserving prolongations for nested bounded-bond families
  and product-flag block merging, with an explicit warning that exact-rank
  tensor manifolds are not nested;
  \item a certificate-driven continuation algorithm that separates encoder
  change, resource enrichment, and local correction;
  \item executed synthetic controls through forty qubits, including matched
  warm/cold work, exact block merging, landmark bounds, and an ordering
  counterexample; and
  \item an explicit boundary between transfer certification and classical
  tractability: when the required overlaps and contractions are available,
  acceptance certifies the declared empirical-objective transfer, whereas
  rejection invalidates only that certificate and does not imply quantum
  advantage.
\end{enumerate}

A necessary geometric distinction is that the full flag manifold in
$\C^{2^n}$ is not a Cartesian product of small local
flag manifolds.  A product of local flags is a restricted, usually separable
ansatz.  Its block subproblems are conditional on all other blocks, not
independent.  Block-coordinate or alternating Riemannian optimization can be
scalable only after this restriction is declared; it does not decompose an
arbitrary global flag.

% Theory and results section included from main.tex.
% Expected preamble: amsmath, amssymb, amsthm, mathtools, bm.
% Expected theorem environments: theorem, proposition, lemma, corollary, remark.
% Expected notation: \C, \tr, \rank, \ket, \bra, \braket, \norm, \abs.

\section{Results}
\label{sec:resource-continuation}

We address when a flag learned at a lower-cost resource level can be
transferred to a richer level without changing the empirical object being
optimized.  Our answer separates three statements that must not be conflated:
(i) transfer of an empirical density between Hilbert spaces, (ii) inclusion of
two structured feasible families at fixed Hilbert space, and (iii) local
geometric stability of individual spectral subspaces.  Only the first two give
objective-value guarantees without an eigengap.

Throughout, $M^\dagger$ denotes the conjugate transpose of $M$,
$I_d$ is the identity on a $d$-dimensional space,
$M\succeq0$ denotes that the Hermitian matrix $M$ is positive semidefinite,
$\operatorname{ran}M$ is its range, and $\lambda_1(M)\ge\lambda_2(M)\ge\cdots$
are its eigenvalues in nonincreasing order.  The symbols $\tr$ and $\rank$
denote trace and rank, respectively.  For vectors, $\norm{\cdot}_2$ is the
Euclidean norm; for matrices, its operator-norm meaning is stated below.

\subsection{Flags, effects, and the two empirical problems}

Let $\mathcal H_c$ and $\mathcal H_f$ be finite-dimensional complex Hilbert
spaces for a coarse and a fine resource level.  Fix ranks
$1\le r_1<\cdots<r_L\le\dim\mathcal H_c$ and nonnegative weights
$\alpha_1,\ldots,\alpha_L$ satisfying
\begin{equation}
  \beta:=\sum_{\ell=1}^{L}\alpha_\ell\le 1.
  \label{eq:rc-weight-normalization}
\end{equation}
A flag $\mathcal F=(P_1,\ldots,P_L)$ consists of Hermitian projectors with
\begin{equation}
 P_\ell P_j=P_{\min\{\ell,j\}},
 \qquad \rank(P_\ell)=r_\ell .
 \label{eq:rc-flag-constraints}
\end{equation}
Its return effect and empirical objective are
\begin{equation}
 A_{\mathcal F}=\sum_{\ell=1}^{L}\alpha_\ell P_\ell,
 \qquad
 \mathcal J_\rho(\mathcal F)=\tr(\rho A_{\mathcal F}).
 \label{eq:rc-flag-objective}
\end{equation}
Nesting and Eq.~\eqref{eq:rc-weight-normalization} imply
$0\preceq A_{\mathcal F}\preceq I$.  We maximize
$\mathcal J_\rho$; all continuation inequalities below reverse in the obvious
way for a minimization convention.

Let $V:\mathcal H_c\to\mathcal H_f$ be an isometry,
$V^\dagger V=I_{\mathcal H_c}$.  It lifts a coarse flag to
\begin{equation}
 V\mathcal F V^\dagger
 :=(VP_1V^\dagger,\ldots,VP_LV^\dagger).
 \label{eq:rc-lifted-flag}
\end{equation}
The lifted operators remain projectors, retain their ranks, and remain nested.
For density operators $\rho_c$ and $\rho_f$, define
\begin{equation}
 \tau:=V\rho_cV^\dagger,
 \qquad
 \varepsilon:=\frac12\norm{\rho_f-\tau}_1 .
 \label{eq:rc-transfer-distance}
\end{equation}
Here and throughout, $\norm{\cdot}_1$ denotes the trace (Schatten-1) norm.
The quantity $\varepsilon$ measures mismatch of the \emph{training empirical
density}; by itself it is not a bound on every unseen query state.

\begin{lemma}[Effect form of trace-distance control]
\label{lem:rc-effect-trace-distance}
For every effect $0\preceq A\preceq I$,
\begin{equation}
 \abs{\tr[(\rho_f-\tau)A]}\le\varepsilon.
 \label{eq:rc-effect-bound}
\end{equation}
Consequently, for every coarse flag $\mathcal F$,
\begin{equation}
 \abs{
 \mathcal J_{\rho_f}(V\mathcal F V^\dagger)
 -\mathcal J_{\rho_c}(\mathcal F)}\le\varepsilon .
 \label{eq:rc-fixed-flag-transfer}
\end{equation}
\end{lemma}

\begin{proof}
Set $D=\rho_f-\tau$ and write its Jordan decomposition as
$D=D_+-D_-$, where $D_\pm\succeq0$ have orthogonal supports.  Since
$\tr D=0$,
\[
 \tr D_+=\tr D_-=\tfrac12\norm D_1=\varepsilon.
\]
For $0\preceq A\preceq I$, both $\tr(D_+A)$ and $\tr(D_-A)$ lie in
$[0,\varepsilon]$, proving Eq.~\eqref{eq:rc-effect-bound}.  For a lifted
flag, use
$A_{V\mathcal F V^\dagger}=VA_{\mathcal F}V^\dagger$ and cyclicity:
\[
 \tr(\tau VA_{\mathcal F}V^\dagger)
 =\tr(V\rho_cA_{\mathcal F}V^\dagger)
 =\tr(\rho_cA_{\mathcal F}).
\]
The lifted flag effect is an effect, so Eq.~\eqref{eq:rc-effect-bound} applies.
\end{proof}

If weights with $\beta>1$ are required, apply the lemma to
$A_{\mathcal F}/\beta$; every occurrence of $\varepsilon$ below becomes
$\beta\varepsilon$.

\begin{remark}[Noisy encoders and general ensembles]
\label{rem:rc-noisy-encoders}
Lemma~\ref{lem:rc-effect-trace-distance} compares two arbitrary fine-level
density operators; purity of the encoded states is never used.  If the fine
encoder is a noisy channel applied to declared inputs, $\rho_f$ is the
empirical average of its outputs and the fixed-flag bound holds with the same
$\varepsilon$; because trace distance contracts under completely positive
trace-preserving maps, any further processing applied identically to both
arms can only decrease $\varepsilon$.  The Gram construction of
Section~\ref{subsec:rc-2m-certificate} extends to finite pure-state ensemble
decompositions by replacing uniform columns $\Phi/\sqrt m$ with weighted
columns $\sqrt{p_i}\,\ket{\psi_i}$.  If the fine and coarse decompositions
contain $m_f$ and $m_c$ components, respectively, the reduced signed matrix
has dimension at most $m_f+m_c$; the $2m$ statement is the equal-size special
case.  This extension still requires complex cross overlaps between the
chosen decomposition vectors, which are not supplied by density matrices or
fidelity-only data alone.  Only the optimum-transfer identity in
Theorem~\ref{thm:rc-flag-transfer}, which uses spectrum preservation under
isometric conjugation, is specific to isometric lifts.
\end{remark}

\begin{theorem}[Certified transfer of a spectral flag]
\label{thm:rc-flag-transfer}
Let
\begin{equation}
 \mathcal J_c^*:=\max_{\mathcal F_c}\mathcal J_{\rho_c}(\mathcal F_c),
 \qquad
 \mathcal J_f^*:=\max_{\mathcal F_f}\mathcal J_{\rho_f}(\mathcal F_f),
 \label{eq:rc-two-optima}
\end{equation}
where the maxima range over all flags with the ranks in
Eq.~\eqref{eq:rc-flag-constraints} on their respective Hilbert spaces.  Then
\begin{equation}
 \abs{\mathcal J_f^*-\mathcal J_c^*}\le\varepsilon.
 \label{eq:rc-optimum-stability}
\end{equation}
If a coarse flag $\widehat{\mathcal F}_c$ is $\delta_c$-suboptimal,
\begin{equation}
 \mathcal J_c^*-
 \mathcal J_{\rho_c}(\widehat{\mathcal F}_c)\le\delta_c,
 \label{eq:rc-coarse-suboptimality}
\end{equation}
then its isometric lift obeys the fine-level certificate
\begin{equation}
 \boxed{
 \mathcal J_f^*-
 \mathcal J_{\rho_f}(V\widehat{\mathcal F}_cV^\dagger)
 \le \delta_c+2\varepsilon .}
 \label{eq:rc-main-certificate}
\end{equation}
No spectral gap is required for
Eqs.~\eqref{eq:rc-optimum-stability}--\eqref{eq:rc-main-certificate}.
\end{theorem}

\begin{proof}
The nonzero eigenvalues of $\tau=V\rho_cV^\dagger$ equal those of $\rho_c$;
the remaining fine-space eigenvalues are zero.  The Ky Fan maximum principle
and nonnegative $\alpha_\ell$ therefore give
\begin{equation}
 \max_{\mathcal F_f}\mathcal J_\tau(\mathcal F_f)
 =\sum_{\ell=1}^{L}\alpha_\ell
   \sum_{j=1}^{r_\ell}\lambda_j(\tau)
 =\mathcal J_c^* .
 \label{eq:rc-tau-coarse-optimum}
\end{equation}
Directions outside $\operatorname{ran}V$ contribute only zero eigenvalues.
Lemma~\ref{lem:rc-effect-trace-distance}, applied to every fine flag effect,
gives
\[
 \mathcal J_\tau(\mathcal F_f)-\varepsilon
 \le\mathcal J_{\rho_f}(\mathcal F_f)
 \le\mathcal J_\tau(\mathcal F_f)+\varepsilon.
\]
Taking maxima and using Eq.~\eqref{eq:rc-tau-coarse-optimum} proves
Eq.~\eqref{eq:rc-optimum-stability}.  Finally,
Lemma~\ref{lem:rc-effect-trace-distance} and
Eq.~\eqref{eq:rc-coarse-suboptimality} yield
\begin{align*}
 \mathcal J_f^*-
 \mathcal J_{\rho_f}(V\widehat{\mathcal F}_cV^\dagger)
 &\le (\mathcal J_c^*+\varepsilon)
      -(\mathcal J_{\rho_c}(\widehat{\mathcal F}_c)-\varepsilon)\\
 &\le\delta_c+2\varepsilon.
\end{align*}
\end{proof}

The factor two in Eq.~\eqref{eq:rc-main-certificate} has a transparent origin:
one perturbation allowance compares the two unrestricted optima and another
compares the transferred coarse flag with its fine score.  The bound is an
empirical training-return statement.  It is not a generalization bound, a test
accuracy guarantee, or a claim that optimization reaches the fine global
optimum.

\begin{remark}[The constant two is attained]
\label{rem:rc-sharpness}
Let $\mathcal H_c=\C^2$ with $\rho_c=I_2/2$, let $V$ embed $\C^2$ into
$\C^3$ as the first two coordinates, and let
$\rho_f=\operatorname{diag}(\tfrac12-\varepsilon,\tfrac12+\varepsilon,0)$
for $\varepsilon\in(0,\tfrac12]$.  Then
$\tfrac12\norm{\rho_f-V\rho_cV^\dagger}_1=\varepsilon$.  For $L=1$,
$r_1=1$, $\alpha_1=1$, every rank-one coarse flag is optimal, so
$\delta_c=0$; if ties are broken adversarially, the returned flag
$\widehat P=e_1e_1^\dagger$ (with $e_1$ the first standard basis vector)
lifts to fine score $\tfrac12-\varepsilon$, while
$\mathcal J_f^*=\tfrac12+\varepsilon$.  The fine suboptimality equals
$2\varepsilon$, so Eq.~\eqref{eq:rc-main-certificate} holds with equality.
The constant cannot be improved without additional assumptions on coarse
spectral gaps or on tie-breaking; degenerate coarse spectra with adversarially
broken ties are the extremal geometry.
\end{remark}

\subsection{A \texorpdfstring{$2m$}{2m}-dimensional amplitude-Gram certificate}
\label{subsec:rc-2m-certificate}

Suppose the same $m$ training samples are encoded at both rungs, with normalized
state columns
\begin{equation}
 \Phi_f=[\ket{\psi^f_1}\ \cdots\ \ket{\psi^f_m}],
 \qquad
 \Phi_c=[\ket{\psi^c_1}\ \cdots\ \ket{\psi^c_m}],
 \label{eq:rc-state-columns}
\end{equation}
and empirical densities
$\rho_f=\Phi_f\Phi_f^\dagger/m$ and
$\rho_c=\Phi_c\Phi_c^\dagger/m$.  Define
\begin{equation}
 X=\frac1{\sqrt m}\begin{bmatrix}\Phi_f&V\Phi_c\end{bmatrix},
 \qquad
 S=\begin{bmatrix}I_m&0\\0&-I_m\end{bmatrix}.
 \label{eq:rc-xs-definitions}
\end{equation}
Then $\rho_f-V\rho_cV^\dagger=XSX^\dagger$.  Its Gram matrix is
\begin{equation}
 K=X^\dagger X
 =\begin{bmatrix}
 G_f&C_{fc}\\ C_{fc}^\dagger&G_c
 \end{bmatrix}\succeq0,
 \label{eq:rc-block-gram}
\end{equation}
where
\begin{equation}
 G_f=\frac1m\Phi_f^\dagger\Phi_f,
 \quad
 G_c=\frac1m\Phi_c^\dagger\Phi_c,
 \quad
 C_{fc}=\frac1m\Phi_f^\dagger V\Phi_c.
 \label{eq:rc-blocks}
\end{equation}
These are complex amplitude overlaps, not squared fidelities.

\begin{theorem}[Exact Gram-implicit transfer certificate]
\label{thm:rc-gram-certificate}
Let $R$ be any Gram factor satisfying $K=R^\dagger R$; for example,
$R=K^{1/2}$.  Set
\begin{equation}
 B:=RSR^\dagger .
 \label{eq:rc-small-difference}
\end{equation}
Then the nonzero eigenvalues of $B$ and
$\rho_f-V\rho_cV^\dagger$ agree, including signs and multiplicities.  Hence
\begin{equation}
 \boxed{
 \varepsilon=\frac12\norm{B}_1,\qquad
 \zeta:=\norm{\rho_f-V\rho_cV^\dagger}_2=\norm{B}_2 .}
 \label{eq:rc-small-norms}
\end{equation}
Here and below, $\norm{\cdot}_2$ denotes the operator (spectral) norm; the
Frobenius norm is not used in this paper.  Both transfer diagnostics are
therefore obtained from a matrix of dimension at most $2m$, without
materializing a vector or operator in $\mathcal H_f$.
\end{theorem}

\begin{proof}
Because $X^\dagger X=R^\dagger R$, the rule
$U(Rz)=Xz$ defines an isometry from $\operatorname{ran}R$ onto
$\operatorname{ran}X$: it is well defined and norm preserving since
\[
 \norm{Rz}^2=z^\dagger Kz=\norm{Xz}^2.
\]
Extend $U$ as a partial isometry on the ambient factor space.  Since
$\operatorname{ran}(RSR^\dagger)\subseteq\operatorname{ran}R$,
\begin{equation}
 XSX^\dagger=U(RSR^\dagger)U^\dagger=UBU^\dagger.
 \label{eq:rc-partial-isometry}
\end{equation}
Isometric conjugation preserves all nonzero singular values; here both sides
are Hermitian, so it also preserves the signed nonzero eigenvalues.  The trace
and operator norm identities follow.
\end{proof}

Equation~\eqref{eq:rc-block-gram} requires fine--coarse cross overlaps
$\bra{\psi_i^f}V\ket{\psi_j^c}$.  If those amplitudes cannot be evaluated
efficiently, the certificate is algebraically valid but not an efficient
algorithm.  A fidelity-only oracle does not determine the complex block Gram
matrix.  With noisy estimated overlaps, Hermitian symmetrization, a declared
positive-semidefinite projection, and propagated numerical uncertainty are
required; Eq.~\eqref{eq:rc-small-norms} is exact only for the matrix actually
factored.

\subsection{Exact and approximate resource prolongation}

The preceding theorem changes the empirical density.  A different and simpler
situation occurs when two resource rungs optimize the same density and the same
objective but use nested feasible families.

\begin{proposition}[Prolongation monotonicity]
\label{prop:rc-prolongation}
Let $\mathcal M_a$ and $\mathcal M_b$ be coarse and fine feasible sets with
objectives $J_a$ and $J_b$, and let
$\mathcal P:\mathcal M_a\to\mathcal M_b$ be a declared prolongation.  Write
$J_a^*=\sup_{\theta\in\mathcal M_a}J_a(\theta)$ and similarly for $J_b^*$.

If prolongation is exact,
\begin{equation}
 J_b(\mathcal P\theta)=J_a(\theta)
 \quad\text{for every }\theta\in\mathcal M_a,
 \label{eq:rc-exact-prolongation}
\end{equation}
then
\begin{equation}
 J_b^*\ge J_a^*.
 \label{eq:rc-optimum-monotonicity}
\end{equation}
If instead it is one-sided $\eta$-accurate,
\begin{equation}
 J_b(\mathcal P\theta)\ge J_a(\theta)-\eta
 \quad\text{for every }\theta\in\mathcal M_a,
 \label{eq:rc-approx-prolongation}
\end{equation}
then
\begin{equation}
 J_b^*\ge J_a^*-\eta.
 \label{eq:rc-approx-optimum}
\end{equation}
Moreover, any fine optimizer that accepts only nondecreasing objective values
returns $\widehat\theta_b$ satisfying
\begin{equation}
 J_b(\widehat\theta_b)
 \ge J_b(\mathcal P\widehat\theta_a)
 \ge J_a(\widehat\theta_a)-\eta.
 \label{eq:rc-algorithm-monotonicity}
\end{equation}
If both rungs share an exact unrestricted upper certificate $C$, their optimal
certificate gaps obey
\begin{equation}
 C-J_b^*\le C-J_a^*+\eta.
 \label{eq:rc-gap-monotonicity}
\end{equation}
\end{proposition}

\begin{proof}
For each $\theta\in\mathcal M_a$,
$J_b^*\ge J_b(\mathcal P\theta)$.  Substitute either
Eq.~\eqref{eq:rc-exact-prolongation} or
Eq.~\eqref{eq:rc-approx-prolongation} and take the supremum over $\theta$.
The accepted-iterate statement follows from monotonicity of the fine solver;
subtracting Eq.~\eqref{eq:rc-approx-optimum} from $C$ proves the final claim.
\end{proof}

At fixed qubit width, site ordering, topology, training density, flag ranks, and
objective, let $\mathcal T_{\le\bm\chi}$ denote the set of tensor trains (or
matrix product states) whose bond ranks are bounded componentwise by
$\bm\chi$.  These cap sets are nested:
\begin{equation}
 \mathcal T_{\le\bm\chi}
 \subseteq \mathcal T_{\le\bm\chi'}
 \quad\text{whenever}\quad
 \chi_j\le\chi'_j\ \text{for every bond }j.
 \label{eq:rc-rank-variety-inclusion}
\end{equation}
Zero-padding the virtual indices supplies an exact representation-level
prolongation, so Proposition~\ref{prop:rc-prolongation} applies with $\eta=0$.
The same argument applies to compatible block refinement or topology enrichment
only when an explicit score-preserving injection is given.  Merely copying
parameters between two unrelated ansatzes is not such an injection.

\begin{remark}[The fixed-rank stratum caveat]
\label{rem:rc-fixed-rank-caveat}
The bounded-rank set $\mathcal T_{\le\bm\chi'}$ is a stratified algebraic
variety, whereas tensors of one \emph{exact} TT-rank $\bm\chi'$ form a smooth
manifold under standard regularity conditions
\citep{holtz_rohwedder_schneider_2012}.  Exact-rank manifolds are not nested:
$\mathcal T_{=\bm\chi}\not\subset\mathcal T_{=\bm\chi'}$ for
$\bm\chi<\bm\chi'$.  A zero-padded rank-$\bm\chi$ tensor lies on a lower-rank
stratum, not at a regular point of the rank-$\bm\chi'$ manifold.  Consequently,
fixed-rank tangent spaces, gauges, and retractions cannot automatically be
evaluated there.  A practical rank-growth step should activate new directions,
for example by a two-site SVD or residual-informed subspace enrichment inspired
by rank-adaptive TT/AMEn methods \citep{dolgov_savostyanov_amen}, followed by a
safeguarded objective check.  Proposition~\ref{prop:rc-prolongation} certifies
the represented warm-start value; it does not prove convergence of that
rank-activation heuristic.
\end{remark}

This distinction favors a fixed-width ladder in which all forty qubits are
present from the first rung and bond caps, block partitions, tree topology, or
flag ranks are enriched afterward.  Such a ladder keeps the empirical target
fixed more often than a direct ``eight-to-forty qubit'' parameter copy.
Flag-rank growth is objective preserving only when all rungs are prefixes of
one declared master flag and the weights on previously present projectors are
unchanged.  Altering the rank list, renormalizing its weights, or increasing the
training/landmark count changes the optimization problem and must instead be
handled by an explicit approximate-prolongation or density-transfer bound.
Riemannian continuation and multilevel optimization provide relevant
predictor--corrector and coarse-correction frameworks
\citep{seguin_kressner_continuation,sutti_vandereycken_multigrid}, but their
convergence hypotheses must be checked for the particular tensor-flag
objective.

\subsection{When qubit-width continuation is exact}

For $s\ge1$, define the ancilla isometry
\begin{equation}
 V_{n\to n+s}: (\C^2)^{\otimes n}\longrightarrow
 (\C^2)^{\otimes(n+s)},
 \qquad
 V_{n\to n+s}\ket\psi=\ket\psi\otimes\ket0^{\otimes s}.
 \label{eq:rc-ancilla-isometry}
\end{equation}

\begin{proposition}[Exact nested-width embedding]
\label{prop:rc-width-embedding}
For every projector $P$ and normalized state $\ket\psi$,
\begin{align}
 \widetilde P
 &:=V_{n\to n+s}PV_{n\to n+s}^\dagger
   =P\otimes\ket{0^s}\!\bra{0^s},
 \label{eq:rc-projector-ancilla}\\
 \bra\psi P\ket\psi
 &=\bigl(\bra\psi\otimes\bra{0^s}\bigr)
    \widetilde P
    \bigl(\ket\psi\otimes\ket{0^s}\bigr).
 \label{eq:rc-score-ancilla}
\end{align}
The map preserves projector rank, flag nesting, amplitude Gram matrices, and all
weighted flag-return scores.  Thus a width continuation is exact whenever the
fine encoder initially satisfies
\begin{equation}
 \ket{\psi_i^f}=V_{n\to n+s}\ket{\psi_i^c}
 \quad\text{for every training sample }i.
 \label{eq:rc-exact-nested-encoder}
\end{equation}
\end{proposition}

\begin{proof}
The identity in Eq.~\eqref{eq:rc-projector-ancilla} follows directly from the
definition of $V_{n\to n+s}$.  Isometric conjugation preserves products and
ranks, hence it preserves projectors and all nesting relations.  Equation
\eqref{eq:rc-score-ancilla} follows from
$V_{n\to n+s}^\dagger V_{n\to n+s}=I$.  Finally,
\[
 \braket{\psi_i^f}{\psi_j^f}
 =\braket{\psi_i^c}{\psi_j^c}\braket{0^s}{0^s}
 =\braket{\psi_i^c}{\psi_j^c},
\]
so the amplitude Gram matrix and every Gram-derived spectral objective are
unchanged.
\end{proof}

The appended qubits initially encode no new information.  They may subsequently
be activated along a declared homotopy, with Theorem~\ref{thm:rc-flag-transfer}
recomputed between accepted rungs.  If the fine encoder instead injects new
features immediately or changes the old gates, Eq.~\eqref{eq:rc-exact-nested-encoder}
does not hold.  Angle copying is then heuristic, and the cross-rung certificate, not the
shared parameter names, determines whether the warm start is defensible.

\subsection{When the flag itself is geometrically stable}

Objective transfer does not imply that individual eigenvectors are stable.
For a requested boundary with $r_\ell<\dim\mathcal H_f$, let $P_\ell^0$ be
the rank-$r_\ell$ leading spectral projector of
$\tau=V\rho_cV^\dagger$, and $P_\ell^f$ the corresponding projector of
$\rho_f$.  Denote the coarse boundary gap, including any zero eigenvalues, by
\begin{equation}
 g_\ell=\lambda_{r_\ell}(\tau)
          -\lambda_{r_\ell+1}(\tau).
 \label{eq:rc-boundary-gap}
\end{equation}

\begin{proposition}[A computable Davis--Kahan acceptance test]
\label{prop:rc-davis-kahan}
Let $\zeta$ be the Gram-implicit operator norm in
Eq.~\eqref{eq:rc-small-norms}.  If $g_\ell>0$ and
$\zeta<g_\ell/2$, then the fine leading subspace has the same dimension and
\begin{equation}
 \norm{P_\ell^f-P_\ell^0}_2
 \le \frac{2\zeta}{g_\ell}.
 \label{eq:rc-davis-kahan-bound}
\end{equation}
Thus a whole flag is accepted geometrically only if the condition holds at
every requested rank boundary.
\end{proposition}

\begin{proof}
Apply Weyl's eigenvalue perturbation bound to
$\rho_f=\tau+(\rho_f-\tau)$.  The assumption
$\norm{\rho_f-\tau}_2=\zeta<g_\ell/2$ prevents the selected and complementary
clusters from crossing the midpoint of the coarse gap.  The Davis--Kahan
$\sin\Theta$ theorem \citep{davis_kahan_1970} with separation at least
$g_\ell/2$ gives
\[
 \norm{(I-P_\ell^0)P_\ell^f}_2
 \le\frac{2\zeta}{g_\ell}.
\]
For equal-rank orthogonal projectors, this largest sine of a principal angle
equals $\norm{P_\ell^f-P_\ell^0}_2$, proving the claim.
\end{proof}

The bound is sufficient rather than necessary.  When $g_\ell=0$, a rank
boundary splits a degenerate cluster and no basis-independent stability claim
for that rank is available.  One should transport the projector onto the whole
cluster or change the rank ladder.  The objective-value guarantee in
Theorem~\ref{thm:rc-flag-transfer} remains valid even at such a degeneracy.
Likewise, warm-starting a partial eigensolver can save iterations when the
subspace is close, but it does not change the cost of one overlap or remove the
need for a gap-dependent residual tolerance; the locally optimal block
preconditioned conjugate gradient (LOBPCG) method is a standard implementation
option \citep{knyazev_lobpcg}.

\subsection{A path-independent Schmidt-rank limitation}

Resource continuation cannot defeat a final ansatz's Schmidt-rank ceiling.
Place $k$ Bell pairs across one declared MPS cut, with any remaining qubits in
$\ket0$:
\begin{equation}
 \ket{\Psi_k}=2^{-k/2}
 \sum_{z\in\{0,1\}^k}\ket z_L\ket z_R.
 \label{eq:rc-cross-cut-bell}
\end{equation}

\begin{theorem}[Path-independent Bell cap]
\label{thm:rc-bell-cap}
For every normalized state $\ket\phi$ whose Schmidt rank across the declared
cut is at most $\chi$,
\begin{equation}
 \abs{\braket{\phi}{\Psi_k}}^2
 \le \min\!\left\{1,\frac{\chi}{2^k}\right\}.
 \label{eq:rc-bell-cap}
\end{equation}
The bound is attained by retaining and renormalizing any
$\min\{\chi,2^k\}$ matched Schmidt terms.  It is therefore also the exact
optimum over open-boundary MPS with global maximum bond $\chi$.  In particular,
no warm start, block-coordinate schedule, Riemannian retraction, or continuation
path ending in that MPS family can exceed Eq.~\eqref{eq:rc-bell-cap}.
\end{theorem}

\begin{proof}
The $2^k$ Schmidt coefficients of $\ket{\Psi_k}$ are all $2^{-k/2}$.  The
Schmidt-form Eckart--Young/Ky Fan variational principle implies that the largest
squared overlap with a normalized rank-$\chi$ vector is the sum of the largest
$\min\{\chi,2^k\}$ squared Schmidt coefficients, giving
Eq.~\eqref{eq:rc-bell-cap}.  The normalized truncation attains equality.  It is
a sum of at most $\chi$ computational-basis product strings, hence has Schmidt
rank at most $\chi$ across every site cut and admits an exact open-boundary MPS
with all bonds at most $\chi$.  Conversely, an MPS bond upper-bounds Schmidt
rank at its associated cut \citep{oseledets_tt,schollwock_mps}.
\end{proof}

For a target infidelity $\nu\in[0,1)$, reaching squared fidelity at least
$1-\nu$ necessarily requires
\begin{equation}
 \chi\ge\left\lceil(1-\nu)2^k\right\rceil.
 \label{eq:rc-bell-required-rank}
\end{equation}
For a $20|20$ partition containing twenty cross-block Bell pairs, two independent
20-qubit factors have $\chi=1$ and fidelity at most $2^{-20}$; exact coupling
through that interface requires $\chi=2^{20}$.  This is not a universal hardness
claim: interleaving paired qubits or changing to a tree aligned with the pairing
can remove that particular large MPS cut.  The theorem certifies the chosen
ordering and topology, consistent with entanglement-based limits on efficient
MPS simulation \citep{vidal_2003}.

\subsection{Certified continuation protocol}

A scientifically controlled forty-qubit study should distinguish the following
rung variables:
\begin{equation}
 s=(n,\bm\chi,\mathcal T,\mathcal B,\bm r,m),
 \label{eq:rc-resource-vector}
\end{equation}
where $n$ is width, $\bm\chi$ the bond profile, $\mathcal T$ the tensor-network
topology and ordering, $\mathcal B$ the block partition, $\bm r$ the flag ranks,
and $m$ the training or landmark count.  A recommended audit is:
\begin{enumerate}
 \item compute the unrestricted amplitude-Gram/Ky--Fan value as the empirical
       teacher certificate;
 \item initialize all forty qubits with a low-cost product or low-bond model;
 \item increase one declared resource at a time using an exact or certified
       approximate prolongation;
 \item activate new tensor directions with residual-informed enrichment and
       accept only nondecreasing true-objective updates;
 \item recompute $\varepsilon$, $\zeta$, boundary gaps, and the unrestricted
       certificate gap at every encoder-changing rung;
 \item compare matched warm and cold starts at the same final feasible family,
       contraction budget, seeds, and train/validation/test discipline.
\end{enumerate}
Warm-start benefit must be reported as total time or contractions to a declared
target, final objective and certificate gap, negative-transfer frequency, and
peak memory.  The cost of all earlier rungs belongs to the warm-start total.

\paragraph{Relation to established components.}
The trace-distance variational identity, Ky Fan extremum, amplitude-Gram duality,
Davis--Kahan perturbation theorem, fixed-rank tensor manifolds, AMEn-style
enrichment, and Riemannian continuation are established ingredients
\citep{horn_johnson,scholkopf_kpca,davis_kahan_1970,
holtz_rohwedder_schneider_2012,dolgov_savostyanov_amen,
seguin_kressner_continuation}.  The Bell cap is a direct specialization of
optimal Schmidt truncation.  Our contribution is their integration into a
Gram-implicit, multirank flag-transfer audit, especially the explicit $2m$
certificate in Theorem~\ref{thm:rc-gram-certificate} combined with
Eq.~\eqref{eq:rc-main-certificate} and the exact/approximate rung distinction.
We do not present warm starting, tensor networks, block-coordinate
optimization, Nystr\"om approximation, or the established perturbation
principles themselves as new.

\subsection{Structured synthetic controls}
\label{sec:results}

Here $t\in[0,1]$ denotes the geodesic activation fraction of each newly added
four-qubit block: $t=0$ is the exact idle-block lift and $t=1$ is full
activation.  A \emph{warm} run inherits the preceding rung's block anchors,
whereas a \emph{cold} run constructs an initialization independently at the
current rung.  Attempted block updates are the primary matched-work unit.
Figure~\ref{fig:continuation-summary} and Table~\ref{tab:key-results}
summarize the executed controls.

\begin{figure}[t]
 \centering
 \includegraphics[width=\textwidth]{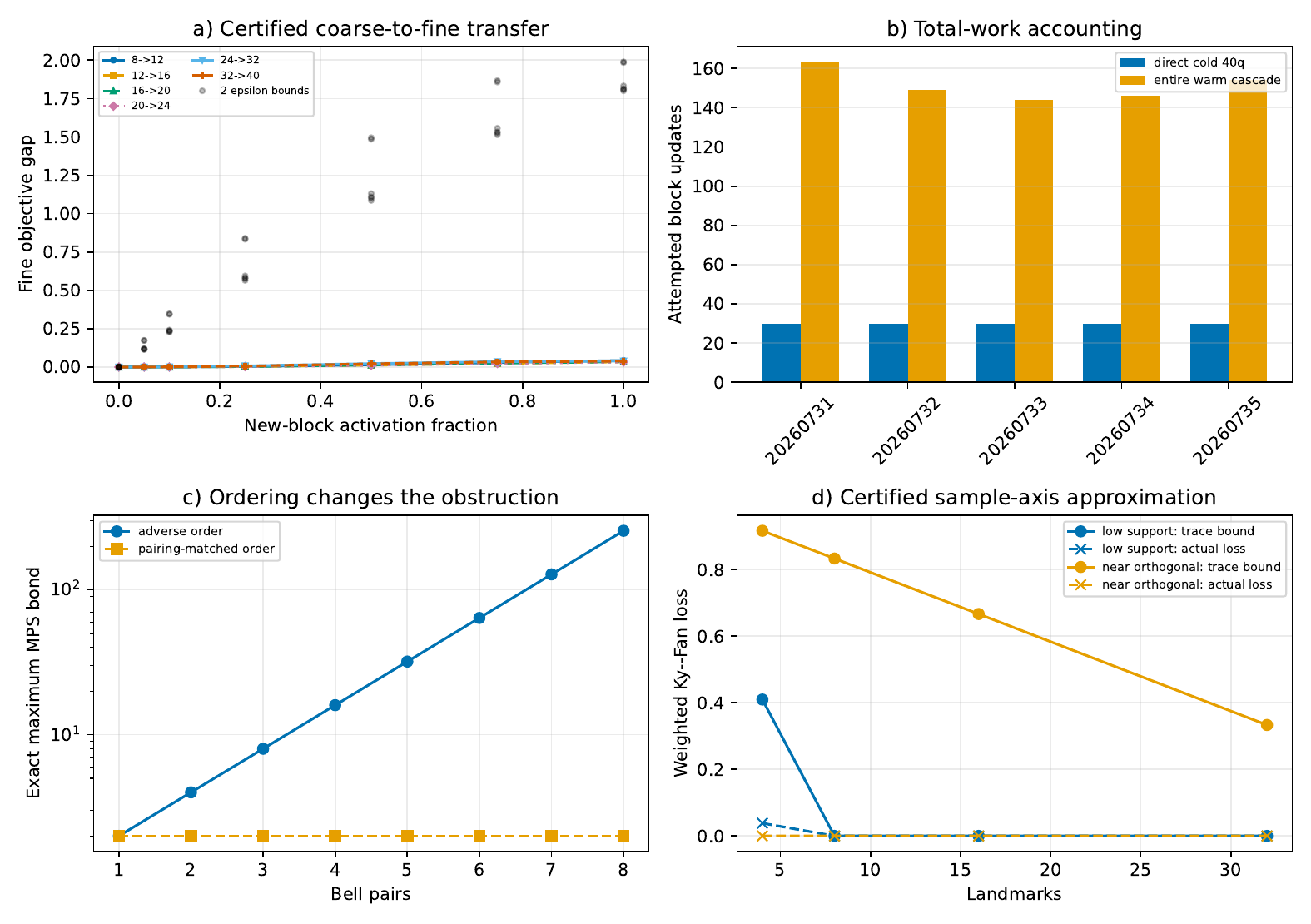}
 \caption{Executed structured synthetic controls.  \textbf{a}, transferred-flag
 objective gaps and their $2\varepsilon$ bounds.  \textbf{b}, complete cascade
 work versus a direct cold 40-qubit fit.  \textbf{c}, exact matrix-product-state
 (MPS) bond dependence on Bell-pair ordering.  \textbf{d}, landmark losses and residual-trace bounds in
 favorable and unfavorable support regimes}
 \label{fig:continuation-summary}
\end{figure}

\begin{table}[t]
\centering
\caption{Executed results and their permissible interpretation.  Values are
computed from the supplied synthetic output summaries; ``no material difference'' means the
absolute warm--cold objective difference was at most $4.04\times10^{-10}$.
Boldface marks the extreme value of each executed range.}
\label{tab:key-results}
\small
\begin{tabularx}{\textwidth}{>{\raggedright\arraybackslash}p{0.24\textwidth}
>{\raggedright\arraybackslash}p{0.28\textwidth}X}
\toprule
Audit & Executed result & Interpretation\\
\midrule
Exact width lift, all six transitions &
 $\varepsilon=7.81$--$\mathbf{8.82\times10^{-16}}$; maximum objective gap
 $\mathbf{8.33\times10^{-17}}$ &
The declared idle-qubit isometry preserves the empirical spectral flag to
rounding.\\
5\% new-block activation &
 $\varepsilon=0.0572$--$\mathbf{0.0872}$; maximum gap $\mathbf{3.20\times10^{-4}}$ &
The bound is below 0.175 but conservative; acceptance still requires a
predeclared tolerance.\\
Full activation &
 $\varepsilon=0.900$--$\mathbf{0.995}$; mean gap $0.0372$ &
The certificate becomes noninformative and should reject an uncertified transition.\\
Final-rung warm versus cold &
20 versus $\mathbf{30}$ block updates in all five seeds; no material objective difference &
Warm initialization saves one final-rung sweep in this controlled family.\\
Complete warm cascade &
 $4.80$--$\mathbf{5.43\times}$ the direct final-rung block-update count &
Final-rung savings do not imply end-to-end savings.\\
Exact 40-qubit block merge &
 maximum preservation residual $\mathbf{3.47\times10^{-18}}$; no material correction gain &
The prolongation is exact, but extra expressivity was unnecessary for this
family.\\
Eight Bell pairs &
 bond $\mathbf{256}$ adverse versus bond 2 pairing-matched &
Topology/order can dominate bond escalation; the reduction factor is 128.\\
Landmark certificate &
exact at eight selected pivots for the embedded eight-dimensional family;
loose for the near-orthogonal family &
Landmark scaling is conditional on effective support, not nominal qubit width.\\
\bottomrule
\end{tabularx}
\end{table}

All 42 transfer rows satisfied both the
fixed-effect $\varepsilon$ inequality and the lifted-optimum $2\varepsilon$
inequality (Fig.~\ref{fig:continuation-summary}a).  At $t=0.05$, mean
$\varepsilon$ was 0.0682 while the mean observed gap was
$2.01\times10^{-4}$.  At full activation, mean $\varepsilon$ rose to 0.936 and
the nominal $2\varepsilon$ bound exceeded one.  This is not a failure of the
theorem: it is the intended signal that a direct coarse-to-fine transfer is no
longer justified and a smaller step, different embedding, or cold restart is
needed.

At the final 40-qubit rung, inheritance increased the initial product-anchor
objective by several orders of magnitude and reduced the completed sweeps from
three to two.  Nevertheless, the seven-rung cascade required a mean 151.2
attempted block updates, versus 30 for the direct final-rung solve, a mean work
ratio of 5.04.  Warm and cold final certificate gaps were both approximately
$2.32\times10^{-6}$ (Fig.~\ref{fig:continuation-summary}b).  The appropriate
conclusion is therefore conditional: this warm start improved the final
initialization and final-rung iteration count, but not total cost.

\section{Discussion}

The evidence supports a hierarchy of decisions rather than a universal
8-to-40-qubit prescription.  When an exact embedding exists, the coarse
solution is an objective-preserving initialization.  When the empirical
density changes, the block Gram supplies an objective-value certificate and,
with an eigengap, a geometric subspace certificate.  When a prespecified
acceptance tolerance is exceeded, additional rungs do not automatically help:
the encoding, ordering, topology, or resource allocation may need to change.

The topology control is especially consequential
(Fig.~\ref{fig:continuation-summary}c).  The same abstract product
of eight Bell pairs has maximum MPS bond 256 when all left partners precede all
right partners, but bond 2 when partners are adjacent.  Train-only correlation
or entanglement graphs and local tree-tensor-network (TTN) reconnections may therefore offer a larger
gain than uniform bond escalation \citep{hikihara_ttn,markov_shi}.  Such a
search is itself combinatorial and can overfit; topology must be selected
without test information.

The principal limitations are:
\begin{itemize}
  \item forty-qubit tractability is conditional on factorization, bounded bond,
  bounded treewidth, or another declared contraction structure;
  \item warm starts reduce search iterations, not the final Schmidt rank or
  the cost of one hard amplitude overlap;
  \item complex amplitude cross overlaps are required; fidelity-only data are
  insufficient;
  \item a bounded-rank set is stratified, so zero padding does not produce a
  regular point of a higher exact-rank manifold;
  \item the Gram/Ky--Fan gap is a training-return certificate, not a
  generalization, anomaly-detection, calibration, or quantum-advantage bound;
  \item eigenvalue degeneracy invalidates basis-vector transport; whole
  projectors or spectral clusters must be tracked;
  \item the landmark trace bound can be very loose for near-orthogonal states;
  and
  \item these experiments are synthetic and exploratory; a fresh
  preregistered confirmatory study remains future work.
\end{itemize}

Optimization on flag manifolds and the need to represent multiple ranks as one
nested object are established \citep{ye_flag,szwagier_pennec_flags}.
Riemannian continuation, multilevel low-rank optimization, fixed
tensor-train (TT)-rank geometry, alternating-minimal-energy (AMEn)
enrichment, MPS/density-matrix-renormalization-group (DMRG) methods, and
randomized Gram approximation are also
established.  Nakamura and Sanji place circuit-cutting overhead inside a
generative eigensolver constraint \citep{nakamura_constrained_gqe}; our use of
explicit bond, block, overlap, and topology budgets follows the same broad
resource-accounting lesson but addresses a different objective and algorithm.

A complementary perspective relates the proposed audits to established
classical-simulability regimes.  Stabilizer circuits, matchgate circuits,
bounded-entanglement and bounded-treewidth states, polynomially generated
Lie-algebra dynamics, and low-rank sample-access models provide distinct
conditions under which quantum calculations admit efficient classical
representations
\citep{aaronson_gottesman_2004,valiant_matchgates_2002,vidal_2003,markov_shi,
goh_gsim,tang_dequantization}.  Related results connect provable trainability
conditions to classical simulability \citep{cerezo_simulability}, construct
classical surrogates of quantum learning models
\citep{schreiber_surrogates,huang_power_of_data}, and give classical
algorithms for specified noisy-circuit observables and sampling regimes
\citep{schuster_noisy}.  The present work does not add a general
simulability theorem.  Conditional on efficient access to the required
structured contractions and complex overlaps, it audits whether a solution
may be transferred between two declared resource models with controlled
empirical-objective loss.  Acceptance certifies that transfer only.  Rejection
means that the chosen map and tolerance do not certify the transfer; it does
not exclude another efficient classical representation and does not imply
quantum advantage.

We therefore do not present warm starting, block-coordinate descent, tensor
networks, Nystr\"om approximation, or Riemannian continuation themselves as
new.  The contribution evaluated here is their integration into a
Gram-implicit flag-transfer audit: the $2m$ cross-rung certificate, its
$\delta_c+2\varepsilon$ consequence, the exact/approximate prolongation
distinction, and matched success and failure controls.

Warm starting can help a forty-qubit tensor-flag computation, but only after
the relationship between rungs is made explicit.  The proposed framework
certifies exact lifts, quantifies approximate transfers from amplitude Grams,
preserves objectives under valid resource enrichment, and rejects paths that
conflict with entanglement or topology limits.  In the executed controls, warm
starting reduced final-rung iterations but increased total work, while topology
changed the required Bell-pair bond by a factor of 128.  The appropriate next
step is therefore not a larger qubit ladder without additional certification.  It is a
prespecified warm/cold-by-topology Pareto study: matched total objective
calls, wall time, and, on hardware, shot budgets, crossed over canonical,
pair-balanced, localized, and adversarially unbalanced partitions, executed
first classically and then at 16--20 qubits on a tensor-network or
quantum-processing-unit (QPU) fragment backend, with adaptive certified steps, residual-based rank
growth, train-only topology selection, and fresh confirmatory data.
Matched-frontier evidence of that kind, rather than a larger qubit label
alone, is what could support a defensible forty-qubit capability claim.

\section{Methods}
\label{sec:methods}

\subsection{Continuation architecture}
\label{sec:architecture}

We use the unrestricted spectral flag as an empirical teacher and a structured
tensor/product flag as a student.  For training columns $\Phi$ and amplitude
Gram $G=\Phi^\dagger\Phi$, the nonzero eigenvalues of the empirical density
$\rho=\Phi\Phi^\dagger/m$ are those of $G/m$.  Hence, for ranks $r_\ell$ and
weights $\alpha_\ell\ge0$, the exact unrestricted certificate is
\begin{equation}
 U(G)=\sum_\ell\alpha_\ell\sum_{j=1}^{r_\ell}\lambda_j(G/m).
 \label{eq:teacher}
\end{equation}
A structured iterate with objective $J$ has the nonnegative training
certificate gap $U(G)-J$.  This is not a generalization bound.

One accepted continuation step is:
\begin{enumerate}
  \item identify the current problem by a hash covering the training data,
  encoder, sample weights, flag ranks, level weights, topology, and ordering;
  \item if the encoder changes, construct the fine/coarse/cross amplitude
  block Gram and compute $\varepsilon$ and the operator-norm perturbation;
  \item reject transfer when the declared objective or geometric acceptance
  threshold is exceeded; otherwise lift the flag and gauge-align its increments;
  \item if only the resource family changes, use an exact prolongation such as
  a larger bounded bond cap or an induced product-block merge;
  \item activate new rank directions by a two-site singular value decomposition
  (SVD) or projected residual,
  rather than relying on zero-padded inactive directions;
  \item perform safeguarded conditional Riemannian or exact local updates,
  accepting only nondecreasing true training objectives; and
  \item record the full cascade cost and compare against cold starts in the
  identical final feasible family.
\end{enumerate}

\subsection{Exact product-block merging}

Let the level-$\ell$ global product projector be
$A_\ell=\bigotimes_bP_{b\ell}$.  Merging adjacent blocks $b,c$ with
\begin{equation}
 P_{bc,\ell}=P_{b\ell}\otimes P_{c\ell},
 \qquad
 \rank(P_{bc,\ell})=r_{b\ell}r_{c\ell},
 \label{eq:block-merge}
\end{equation}
leaves $A_\ell$ and every return exactly unchanged.  The merged projector may
then move away from the Kronecker submanifold in the larger local Grassmann or
flag space.  For multilevel flags, the merged Stiefel columns must be ordered by
the first level at which each tensor-product column pair is present; unadjusted
lexicographic Kronecker ordering need not make every product subspace a prefix.

\subsection{Certified landmark flags}

Once qubit contractions are controlled, the $O(m^2)$ overlap count and
$O(m^3)$ dense eigendecomposition become the next bottleneck.  Let $S$ index
the selected landmark samples, let $\Phi_{:S}$ collect their state columns,
and let $Q$ be an orthonormal basis of $\operatorname{ran}\Phi_{:S}$.  With
$G_{SS}^{+}$ the Moore--Penrose pseudoinverse of the landmark block, define
\begin{equation}
 \widetilde G=G_{:S}G_{SS}^{+}G_{S:}.
\end{equation}
Then
\begin{equation}
 G-\widetilde G
 =\Phi^\dagger(I-QQ^\dagger)\Phi\succeq0.
\end{equation}
With $\eta=\tr(G-\widetilde G)/m$, Ky Fan monotonicity gives
\begin{equation}
 0\le \sum_{j=1}^{r}\lambda_j(G/m)
       -\sum_{j=1}^{r}\lambda_j(\widetilde G/m)
 \le\eta.
 \label{eq:landmark-bound}
\end{equation}
Thus a weighted flag with level-weight sum at most one loses at most $\eta$ in
the unrestricted training objective.  Randomly pivoted Cholesky supplies
strong entry-efficient guarantees in favorable kernel regimes
\citep{chen_random_cholesky}; our executed reference uses deterministic maximum
residual-diagonal pivots with the full Gram already supplied and therefore
claims no entry-oracle speedup.

\subsection{Structured synthetic evaluation}
\label{sec:protocol}

All reported results were produced by a deterministic internal research
implementation using the definitions and protocols described here.  The study is
synthetic, classical, and structured.  No dense 40-qubit vector, hardware
execution, molecular energy, or real-data test metric was used.

\paragraph{Synthetic witness.}
Each experiment used $m=48$ samples, $B=10$ four-qubit blocks, and local
dimension $d=16$.  Balanced generation labels $y_i\in\{0,1\}$ (24 per value)
were used only to construct the witness; no classifier or test set was used.
For sample $i$ and zero-indexed block $b$, independent real vectors
$\xi^{(R)}_{ib},\xi^{(I)}_{ib}\sim\mathcal N(0,I_{16})$ defined
\begin{equation}
 v_{ib}=0.28\bigl(\xi^{(R)}_{ib}+\mathrm{i}\xi^{(I)}_{ib}\bigr)
       +0.8\,e_{(y_i+b)\bmod 2},
 \qquad
 \ket{\psi_{ib}}=v_{ib}/\norm{v_{ib}}_2,
 \label{eq:synthetic-witness}
\end{equation}
where $e_j$ is the $j$th standard basis vector in $\C^{16}$.  Samples were
permuted after generation.  The implicit global state was
$\ket{\psi_i}=\bigotimes_{b=0}^{B-1}\ket{\psi_{ib}}$ and was never
materialized; global amplitude-Gram entries were evaluated as
$G_{ij}=\prod_b\braket{\psi_{ib}}{\psi_{jb}}$.  A PCG64 pseudorandom generator used seed 20260730 for the
transfer lane, seeds 20260731--20260735 for matched warm/cold runs, and seed
20260736 for the low-support landmark control.

\paragraph{Transfer lane.}
We exposed the fixed ladder $8,12,16,20,24,32,40$.  For each target block
$q$, its global phase was fixed so that $c=\langle e_0,q\rangle\ge0$ was real.
Writing $\theta=\arccos(c)$ and
$u=(q-ce_0)/\sin\theta$ when $\sin\theta>0$, the activation path was
$q(t)=\cos(t\theta)e_0+\sin(t\theta)u$; the degenerate case used $q(t)=e_0$.
We evaluated
$t\in\{0,0.05,0.10,0.25,0.50,0.75,1\}$.  For every transition we computed the
fine, coarse, and cross amplitude Grams, the exact transfer certificate, and
rank-$(1,2,4)$ spectral flags with weights $(0.5,0.3,0.2)$.

\paragraph{Warm/cold optimization lane.}
For seeds 20260731--20260735 we optimized one normalized anchor $a_b\in\C^{16}$
per active block by maximizing
\begin{equation}
 J(a_1,\ldots,a_{B'})=\frac1m\sum_{i=1}^{m}
   \prod_{b=1}^{B'}\abs{\langle a_b,\psi_{ib}\rangle}^2,
 \label{eq:product-anchor-objective}
\end{equation}
using exact conditional block updates.  The cold arm initialized $a_b$ as the
leading eigenvector of the local empirical density
$m^{-1}\sum_i\ket{\psi_{ib}}\!\bra{\psi_{ib}}$.  The warm arm inherited every
old anchor and applied the same local initialization only to newly introduced
blocks.  Both arms stopped when one sweep improved $J$ by at most $10^{-12}$
or after 20 sweeps.  Attempted block updates are the primary work unit; wall
time is descriptive.

\paragraph{Resource and failure controls.}
For the first three warm/cold seeds, at 40 qubits we exactly merged the first
two four-qubit blocks into one eight-qubit block and then applied safeguarded
corrections.  We represented
one through eight Bell pairs under both a central cross-cut ordering and an
adjacent-pair ordering.  Finally, we evaluated the landmark bound on both the
near-orthogonal 40-qubit product family and a 40-qubit family supported on an
embedded eight-dimensional active subspace.  The latter was generated by
normalizing 48 independent complex Gaussian columns in $\C^8$ using seed
20260736.  Landmark counts were $4,8,16,$ and $32$, with deterministic
maximum-residual-diagonal pivoting applied to the supplied full Gram matrix
(Fig.~\ref{fig:continuation-summary}d).

\paragraph{Confirmatory status.}
The results reported here are the executed synthetic exploratory study and
the stated deterministic controls.  A fresh preregistered confirmatory study
with separately fixed seeds and analysis decisions remains future work; no
unexecuted registry row is presented or distributed as evidence.

\subsection{Computational and writing assistance}

OpenAI ChatGPT/Codex and Anthropic Claude were used as tools for mathematical
and software stress-testing, code review, and language assistance.  They were
not treated as authors or as sources of scientific evidence.  The human
authors selected the research questions and methods, inspected and corrected
tool outputs, executed and verified the calculations, and verified every
theorem, proof, numerical result, citation, and claim.  The human authors
retain full responsibility for the work.

\section*{Data availability}

All study data are synthetic.  The six machine-readable CSV files underlying
Fig.~\ref{fig:continuation-summary}, Table~\ref{tab:key-results}, and every
reported numerical summary are supplied as Online Resource~1.  They contain
136 data rows and no source code, credentials, customer data, or private
operating rules.  The generation procedure, seeds, sample count, qubit ladder,
flag ranks and weights, stopping rule, and primary work unit are specified in
Methods.  No human, animal, clinical, or hardware-provider data were used.

\section*{Code availability}

The implementation and reproduction code is company-owned research software
of Pattern Recognition Pty Ltd and is not publicly released.  To permit
article-level verification without disclosing proprietary implementation,
Online Resource~1 supplies all numerical outputs needed to reconstruct the
reported figure, table, and summaries, while the manuscript provides the
mathematical definitions, theorem statements, and complete experimental
protocol.  No production scheduler, provider credentials, hardware-specific
configuration, customer data, or unexecuted protocol registry is released.

\section*{Supplementary information}

Online Resource~1: six machine-readable CSV files containing the 136 synthetic
output rows underlying Fig.~\ref{fig:continuation-summary},
Table~\ref{tab:key-results}, and all reported numerical summaries.

\section*{Acknowledgements}

The authors acknowledge assistance from OpenAI ChatGPT/Codex and Anthropic Claude with language refinement, critical examination of manuscript arguments, and software review. Any resulting suggestions incorporated into the manuscript were evaluated and, where relevant, independently verified by the authors, who retain full responsibility for the work.

\section*{Statements and Declarations}

\subsection*{Funding}

No external funding was received for this study.

\subsection*{Ethics approval}

Not applicable.  No human participants, animals, or personal data were
involved in this study.

\subsection*{Author contributions}

Azadeh Alavi (A.Az.), Fatemeh Kouchmeshki (F.K.), and Hossein Akhoundi (H.A.)
contributed equally to conceptualization, methodology, formal analysis,
software, validation, investigation, and drafting and revision of the
manuscript.  A.Az. coordinated the project.  Abdolrahman Alavi (A.Ab.) provided
supervision and critical review of the manuscript.  All authors approved the
submitted version and accept responsibility for the integrity of the work.

\subsection*{Competing interests}

The authors and Pattern Recognition Pty Ltd have proprietary and commercial
interests in geometry-aware quantum-representation and resource-continuation
software.  An Australian provisional patent application has been
filed; A.Az., F.K., and H.A. are named inventors.  A.Az. is affiliated with
RMIT University and Pattern Recognition Pty Ltd; the other authors are
affiliated with Pattern Recognition Pty Ltd.  These interests did not alter
the study design, the complete reporting of negative results, or the stated
limitations.

\bibliographystyle{plainnat}
\bibliography{theory_references,references}

@article{aaronson_gottesman_2004,
  author  = {Scott Aaronson and Daniel Gottesman},
  title   = {Improved Simulation of Stabilizer Circuits},
  journal = {Physical Review A},
  volume  = {70},
  number  = {5},
  pages   = {052328},
  year    = {2004},
  doi     = {10.1103/PhysRevA.70.052328}
}

@article{valiant_matchgates_2002,
  author  = {Leslie G. Valiant},
  title   = {Quantum Circuits That Can Be Simulated Classically in Polynomial Time},
  journal = {SIAM Journal on Computing},
  volume  = {31},
  number  = {4},
  pages   = {1229--1254},
  year    = {2002},
  doi     = {10.1137/S0097539700377025}
}

@article{goh_gsim,
  author  = {Matthew L. Goh and Martin Larocca and Lukasz Cincio and M. Cerezo and Fr{\'e}d{\'e}ric Sauvage},
  title   = {Lie-Algebraic Classical Simulations for Quantum Computing},
  journal = {Physical Review Research},
  volume  = {7},
  pages   = {033266},
  year    = {2025},
  doi     = {10.1103/PhysRevResearch.7.033266},
  url     = {https://doi.org/10.1103/PhysRevResearch.7.033266}
}

@inproceedings{tang_dequantization,
  author    = {Ewin Tang},
  title     = {A Quantum-Inspired Classical Algorithm for Recommendation Systems},
  booktitle = {Proceedings of the 51st Annual ACM SIGACT Symposium on Theory of Computing (STOC)},
  pages     = {217--228},
  year      = {2019},
  doi       = {10.1145/3313276.3316310}
}

@article{huang_power_of_data,
  author  = {Hsin-Yuan Huang and Michael Broughton and Masoud Mohseni and Ryan Babbush and Sergio Boixo and Hartmut Neven and Jarrod R. McClean},
  title   = {Power of Data in Quantum Machine Learning},
  journal = {Nature Communications},
  volume  = {12},
  pages   = {2631},
  year    = {2021},
  doi     = {10.1038/s41467-021-22539-9}
}

@article{cerezo_simulability,
  author  = {M. Cerezo and Martin Larocca and Diego Garc{\'i}a-Mart{\'i}n and N. L. Diaz and Paolo Braccia and Enrico Fontana and Manuel S. Rudolph and Pablo Bermejo and Aroosa Ijaz and Supanut Thanasilp and Eric R. Anschuetz and Zo{\"e} Holmes},
  title   = {Does Provable Absence of Barren Plateaus Imply Classical Simulability?},
  journal = {Nature Communications},
  volume  = {16},
  pages   = {7907},
  year    = {2025},
  eprint  = {2312.09121},
  archivePrefix = {arXiv},
  primaryClass  = {quant-ph},
  doi     = {10.1038/s41467-025-63099-6}
}

@misc{schuster_noisy,
  author        = {Thomas Schuster and Chao Yin and Xun Gao and Norman Y. Yao},
  title         = {A Polynomial-Time Classical Algorithm for Noisy Quantum Circuits},
  year          = {2024},
  eprint        = {2407.12768},
  archivePrefix = {arXiv},
  primaryClass  = {quant-ph},
  doi           = {10.48550/arXiv.2407.12768},
  url           = {https://arxiv.org/abs/2407.12768}
}

@article{schreiber_surrogates,
  author  = {Franz J. Schreiber and Jens Eisert and Johannes Jakob Meyer},
  title   = {Classical Surrogates of Quantum Learning Models},
  journal = {Physical Review Letters},
  volume  = {131},
  number  = {10},
  pages   = {100803},
  year    = {2023},
  doi     = {10.1103/PhysRevLett.131.100803}
}

@article{chen_random_cholesky,
  author        = {Yifan Chen and Ethan N. Epperly and Joel A. Tropp and Robert J. Webber},
  title         = {Randomly Pivoted {C}holesky: Practical Approximation of a Kernel Matrix with Few Entry Evaluations},
  journal       = {Communications on Pure and Applied Mathematics},
  volume        = {78},
  number        = {5},
  pages         = {995--1041},
  year          = {2025},
  doi           = {10.1002/cpa.22234},
  url           = {https://doi.org/10.1002/cpa.22234}
}

@article{hikihara_ttn,
  author  = {Toshiya Hikihara and Hiroshi Ueda and Kouichi Okunishi and Kenji Harada and Tomotoshi Nishino},
  title   = {Automatic Structural Optimization of Tree Tensor Networks},
  journal = {Physical Review Research},
  volume  = {5},
  pages   = {013031},
  year    = {2023},
  doi     = {10.1103/PhysRevResearch.5.013031},
  url     = {https://doi.org/10.1103/PhysRevResearch.5.013031}
}

@article{markov_shi,
  author  = {Igor L. Markov and Yaoyun Shi},
  title   = {Simulating Quantum Computation by Contracting Tensor Networks},
  journal = {SIAM Journal on Computing},
  volume  = {38},
  number  = {3},
  pages   = {963--981},
  year    = {2008},
  doi     = {10.1137/050644756}
}

@article{nakamura_constrained_gqe,
  author        = {Junya Nakamura and Shinichiro Sanji},
  title         = {Generative Quantum Eigensolver with Constrained Circuit-Cutting Overhead},
  journal       = {arXiv preprint arXiv:2509.08351},
  year          = {2025},
  note          = {Version 3, revised 10 May 2026},
  eprint        = {2509.08351},
  archivePrefix = {arXiv},
  primaryClass  = {quant-ph},
  url           = {https://arxiv.org/abs/2509.08351}
}

@misc{szwagier_pennec_flags,
  author        = {Tom Szwagier and Xavier Pennec},
  title         = {Nested Subspace Learning with Flags},
  year          = {2025},
  eprint        = {2502.06022},
  archivePrefix = {arXiv},
  primaryClass  = {stat.ML},
  doi           = {10.48550/arXiv.2502.06022},
  url           = {https://arxiv.org/abs/2502.06022}
}

@misc{alavi_invariance_2026,
  author        = {Azadeh Alavi and Fatemeh Kouchmeshki and Hossein Akhoundi},
  title         = {Invariance Audits for Quantum Kernels and Variational Rewinding: A Real-to-Hermitian Taxonomy of Projector, Flag, Anchor, and Density Geometry},
  year          = {2026},
  eprint        = {2607.07927},
  archivePrefix = {arXiv},
  primaryClass  = {quant-ph},
  doi           = {10.48550/arXiv.2607.07927},
  url           = {https://arxiv.org/abs/2607.07927}
}

@article{davis_kahan_1970,
  author  = {Chandler Davis and W. M. Kahan},
  title   = {The Rotation of Eigenvectors by a Perturbation. {III}},
  journal = {SIAM Journal on Numerical Analysis},
  volume  = {7},
  number  = {1},
  pages   = {1--46},
  year    = {1970},
  doi     = {10.1137/0707001}
}

@article{dolgov_savostyanov_amen,
  author  = {Sergey V. Dolgov and Dmitry V. Savostyanov},
  title   = {Alternating Minimal Energy Methods for Linear Systems in Higher Dimensions. {Part I}: {SPD} Systems},
  journal = {SIAM Journal on Scientific Computing},
  volume  = {36},
  number  = {5},
  pages   = {A2248--A2271},
  year    = {2014},
  doi     = {10.1137/140953289},
  url     = {https://doi.org/10.1137/140953289}
}

@article{holtz_rohwedder_schneider_2012,
  author  = {Sebastian Holtz and Thorsten Rohwedder and Reinhold Schneider},
  title   = {On Manifolds of Tensors of Fixed {TT}-Rank},
  journal = {Numerische Mathematik},
  volume  = {120},
  number  = {4},
  pages   = {701--731},
  year    = {2012},
  doi     = {10.1007/s00211-011-0419-7}
}

@book{horn_johnson,
  author    = {Roger A. Horn and Charles R. Johnson},
  title     = {Matrix Analysis},
  edition   = {Second},
  publisher = {Cambridge University Press},
  address   = {Cambridge},
  year      = {2013},
  doi       = {10.1017/CBO9781139020411}
}

@article{knyazev_lobpcg,
  author  = {Andrew V. Knyazev},
  title   = {Toward the Optimal Preconditioned Eigensolver: Locally Optimal Block Preconditioned Conjugate Gradient Method},
  journal = {SIAM Journal on Scientific Computing},
  volume  = {23},
  number  = {2},
  pages   = {517--541},
  year    = {2001},
  doi     = {10.1137/S1064827500366124}
}

@article{oseledets_tt,
  author  = {Ivan V. Oseledets},
  title   = {Tensor-Train Decomposition},
  journal = {SIAM Journal on Scientific Computing},
  volume  = {33},
  number  = {5},
  pages   = {2295--2317},
  year    = {2011},
  doi     = {10.1137/090752286}
}

@article{schollwock_mps,
  author  = {Ulrich Schollw{\"o}ck},
  title   = {The Density-Matrix Renormalization Group in the Age of Matrix Product States},
  journal = {Annals of Physics},
  volume  = {326},
  number  = {1},
  pages   = {96--192},
  year    = {2011},
  doi     = {10.1016/j.aop.2010.09.012}
}

@article{scholkopf_kpca,
  author  = {Bernhard Sch{\"o}lkopf and Alexander Smola and Klaus-Robert M{\"u}ller},
  title   = {Nonlinear Component Analysis as a Kernel Eigenvalue Problem},
  journal = {Neural Computation},
  volume  = {10},
  number  = {5},
  pages   = {1299--1319},
  year    = {1998},
  doi     = {10.1162/089976698300017467}
}

@misc{seguin_kressner_continuation,
  author        = {Axel S{\'e}guin and Daniel Kressner},
  title         = {Continuation Methods for {R}iemannian Optimization},
  year          = {2021},
  eprint        = {2106.08839},
  archivePrefix = {arXiv},
  primaryClass  = {math.OC},
  doi           = {10.48550/arXiv.2106.08839},
  url           = {https://arxiv.org/abs/2106.08839}
}

@article{sutti_vandereycken_multigrid,
  author  = {Marco Sutti and Bart Vandereycken},
  title   = {Riemannian Multigrid Line Search for Low-Rank Problems},
  journal = {SIAM Journal on Scientific Computing},
  volume  = {43},
  number  = {3},
  pages   = {A1803--A1831},
  year    = {2021},
  doi     = {10.1137/20M1337430},
  url     = {https://doi.org/10.1137/20M1337430}
}

@article{vidal_2003,
  author  = {Guifr{\'e} Vidal},
  title   = {Efficient Classical Simulation of Slightly Entangled Quantum Computations},
  journal = {Physical Review Letters},
  volume  = {91},
  number  = {14},
  pages   = {147902},
  year    = {2003},
  doi     = {10.1103/PhysRevLett.91.147902}
}

@article{ye_flag,
  author  = {Ke Ye and Ken Sze-Wai Wong and Lek-Heng Lim},
  title   = {Optimization on Flag Manifolds},
  journal = {Mathematical Programming},
  volume  = {194},
  number  = {1--2},
  pages   = {621--660},
  year    = {2022},
  doi     = {10.1007/s10107-021-01640-3}
}

\end{document}